\documentclass[12pt]{amsart}

\usepackage{amsmath,amsthm,amssymb,amsfonts,verbatim,multirow}
\usepackage[colorlinks,
            linkcolor=blue,
            anchorcolor=blue,
            citecolor=blue
            ]{hyperref}
\usepackage[hmargin=1.2in,vmargin=1.2in]{geometry}

\title[Bounds and Constructions for Insertion and Deletion Codes]{Bounds and Constructions for Insertion and Deletion Codes}

\author{Shu Liu}\address{National Key Laboratory of Science and Technology on Communications, University of Electronic Science and Technology of China, Chengdu, China} \email{shuliu@uestc.edu.cn}
\author{Chaoping Xing} \address{School of Electronic Information and Electric Engineering, Shanghai Jiao Tong University, Shanghai, China} \email{xingcp@sjtu.edu.cn}

\date{}

\date{}

\newtheorem{lemma}{Lemma}[section]
\newtheorem{theorem}[lemma]{Theorem}
\newtheorem{cor}[lemma]{Corollary}

\newtheorem{defn}{Definition}

\newtheorem{open}{Open Problem}

\theoremstyle{remark}
\newtheorem{rmk}{Remark}

\renewcommand{\epsilon}{\varepsilon}
\renewcommand{\le}{\leqslant}
\renewcommand{\ge}{\geqslant}

\newcommand{\vnote}[1]{}

\def\ZZ{\mathbb{Z}}
\def\PP{\mathbb{P}}

\def\F{\mathbb{F}}
\def \mC {\mathcal{C}}
\def \mA {\mathcal{A}}

\def \mA {\mathcal{A}}

\def \mC {\mathcal{C}}

\def \Xi {{X^{[i]}}}

\newcommand{\Ga}{\alpha}
\newcommand{\Gb}{\beta}
     
\newcommand{\Gd}{\delta}

\newcommand{\Gs}{\sigma}

\def\bGa {{\bf \Ga}}
\def \bi {{\bf i}}
\def \bj {{\bf j}}
\def \ba {{\bf a}}
\def \bb {{\bf b}}
\def \bc {{\bf c}}
\def \bi {{\bf i}}
\def \bx {{\bf x}}

\def \bu {{\bf u}}
\def \bv {{\bf v}}
\def \bo {{\bf 0}}

\def\supp {{\rm supp }}

\def\Aut {{\rm Aut }}

\def\AGL {{\rm AGL}}

\def\wt{{\rm wt}}

\def\LCS{{\rm LCS}}
\def\RS{{\textsf{RS} }}

\def\wt{{\rm wt}}

\begin{document}

\maketitle

\begin{abstract} Insertion and deletion (insdel for short) codes have recently attracted a lot of attention due to their applications in many interesting fields such as DNA storage, DNA analysis, race-track memory error correction and language processing. The present paper mainly studies limits and constructions of insdel codes. The paper can be divided into two parts. The first part focuses on  various bounds, while the second part concentrates on constructions of insdel codes.

Although the insdel-metric Singleton bound has been derived before, it is still unknown if there are any nontrivial codes achieving this bound. Our first result shows that any nontrivial insdel codes do not achieve the insdel-metric Singleton bound. The second bound shows that every $[n,k]$ Reed-Solomon code has insdel distance upper bounded by $2n-4k+4$ and it is known in literature that an $[n,k]$ Reed-Solomon code can have insdel distance $2n-4k+4$ as long as the field size is sufficiently large. The third bound shows a trade-off between insdel distance and code alphabet size for  codes achieving the Hamming-metric Singleton bound. In the second part of the paper, we first provide a non-explicit construction of nonlinear codes that can approach the insdel-metric Singleton bound arbitrarily when the code alphabet size is sufficiently large. The second construction gives two-dimensional Reed-Solomon codes of length $n$ and insdel distance $2n-4$ with field size $q=O(n^5)$.

The  non-explicit construction of insdel codes is based on constant-weight $L^1$-codes that are  introduced in this paper. We first establish a relation between constant-weight $L^1$-codes and insdel codes. Based on this relation, we construct constant-weight $L^1$-codes with reasonable parameters and subsequently give insdel codes approaching the insdel-metric Singleton bound. Via automorphism group of rational function field, we provide a necessary and sufficient condition under which a two-dimensional Reed-Solomon code of length $n$ has insdel distance $2n-4$. Based on this criterion, we present a construction of $q$-ary two-dimensional Reed-Solomon codes of length $n$ and insdel distance $2n-4$ with $q=O(n^5)$. Though this is worse than the current best field size, we provide a new angle to look into the problem.

\end{abstract}

\section{Introduction}
Classical error-correcting codes under the Hamming metric are widely used to correct substitution and erasure errors. A different class of codes, called insertion and deletion (insdel for short) codes, are designed to correct synchronization errors~\cite{HS2017,HSS2018}  in communication systems caused by the loss of positional information of the message. Insdel codes have recently attracted lots of attention due to their applications in many interesting fields such as DNA storage, DNA analysis~\cite{JHSB2017,XW2005}, race-track memory error correction \cite{CKVVY2017} and language processing \cite{BM2000,O2003}. The study of codes with insdel errors were pioneered by Levenshtein, Varshamov and Tenengolts in the 1960s~\cite{VT65,L1965,L1967,T1984}. Afterwards, various aspects of insdel codes such as bounds, constructions and decoding algorithms have been studied in literatures (see \cite{SY2002,B1995,M1998,Y2001,SWY2002,TS2007,AFC2007,BGZ2018,HS2021}).

In the insdel setting, the distance between two words is the smallest number of insertions and deletions needed to transform one codeword into the other codeword. The minimum insdel distance of a code is the minimum insdel distance among its codewords. %It is an important parameter which shows its insdel error-correcting capability.
Like codes under other metrics, we are interested in  codes with large minimum distance as well as large code size. As one can imagine, similar to the other metrics, there is a trade-off between code size and insdel distance. Thus, optimizing this trade-off is one of the central problems in the topic of insdel codes.

%It would certainly be great if the insdel code size and the minimum insdel distance as large as possible for a fixed code length. Similar to the Hamming metric case, there is a trade-off between code size and insdel distance. It is natural to consider the problems of providing bounds for insdel codes and constructing insdel codes achieving bounds with equality.

Reed-Solomon codes are the most widely used family of codes in both theory and practice. Study of Reed-Solomon codes under insdel-metric has received great attention in the last few years \cite{DLTX2021,CZ2021,CST2021,WMS2004,MS2007,LT2021,TS2007}. Although Reed-Solomon codes achieve the best trade-off between error correction and rate under the Hamming metric, one wonders if this class of codes also has good performance under the insdel-metric.

\subsection{Known results}
There have been a few constructions of insdel codes  in early literatures. Sloane~\cite{S2002}  gave constructions of single deletion correcting codes. He employed the exhaustive search and reported the largest single deletion correcting binary codes of various given lengths. He also showed that the Varshamov-Tenengolts codes~\cite{VT65} are capable of correcting one deletion. Bours~\cite{B1995} constructed  $2$-deletion correcting codes of length $4$ and $3$-deletion correcting codes of length $5$ via combinatorial designs. Furthermore, Mahmoodi~\cite{M1998} gave more constructions for $3$-deletion correcting codes of length $5$. Yin~\cite{Y2001} and Shalaby et al.~\cite{SWY2002} provided constructions of $4$-deletion correcting codes of length $6$.

The Singleton bound for insdel-metric states that a $q$-ary $(n,M)$-code with insdel distance $d$ must obey $M\le q^{n-\frac d2+1}$ or $k\le n-\frac d2+1$ with $k=\log_qM$. So far, except for some trivial codes such as $d=2$ or $d=2n$, one does not know if there are any insdel codes achieving the Singleton bound. This gives the following open problem.
\begin{open}
Are there any $q$-ary $(n,M)$-codes with insdel distance $d$ achieving $M= q^{n-\frac d2+1}$ for $2<d<2n$?
\end{open}
On the other hand, Levenshtein \cite{L2002} derived a lower bound stating that there exists a $q$-ary $(n,M)$-code of insdel distance $d$ with
\[\log_qM\ge n-\frac d2-O\left(\frac{n\log n}{\log q}\right).\]
 This means that we have $q$-ary $(n,M)$-codes of insdel distance $d$ with $k=\log_qM$ arbitrarily approaching $n-\frac d2$ when $q$ is large enough with respect to the length $n$. In view of this result, the following open problem arises.
\begin{open}
Are there any $q$-ary $(n,M)$-codes of insdel distance $d$ with $k=\log_qM$ arbitrarily approaching $n-\frac d2+1$ when $q$ is large enough with respect to the length $n$?
\end{open}
In the above discussions, we assume that the code alphabet size $q$ is sufficiently large. If we consider the typical asymptotical scenario where the code alphabet size $q$ is fixed and the code length $n$ tends to infinity, then by using the Synchronization techniques Haeupler and  Shahrasbi showed in \cite{HS2017} that there are $q$-ary $(n,M)$-codes of insdel distance $d$ with
\[\log_qM=n-\frac d2-O\left(\frac{1}{q}\right)n.\]
However, if we come to linear codes, the bound is suddenly reduced to half. More precisely speaking, it was shown in \cite{CGHL2021} that any $q$-ary $[n,k]$-linear code with insdel distance $d$ must obey
\[k\le \frac12\left(n-\frac d2\right)+o(n).\]
This means that linear codes can only achieve half of the Singleton bound though nonlinear codes can arbitrarily approach the Singleton bound.

So far, we have gotten the sense that the bigger code alphabet size is, the larger insdel distance could be, i.e., there is a trade-off between insdel distance and code alphabet size. There is some implicit research on this problem for general codes \cite{L2002,HS2017,CGHL2021}. For Reed-Solomon codes, the problem was explicitly investigated in the paper \cite{CST2021} where they show some lower and  upper bounds on code alphabet size for given minimum insdel distance.  More precisely speaking, it was shown in \cite{CST2021} that a $q$-ary $[n,k]$ Reed-Solomon code with insdel distance $2n-4k+4$ has the smallest alphabet size $q$ satisfying
\[\Omega\left(\left(\frac{n}{k^2}\right)^{\frac{2k-1}{k-1}}\right)\le q\le O\left(n^{4k-2}\right).\]
Study of the trade-off between insdel  distance and code alphabet size for other classes of codes such as Hamming-metric Singleton-optimal codes has not been done in literature (here we refer Hamming-metric Singleton-optimal codes to those achieving the Singleton bound under the Hamming metric).

%%%%%%%%%%%%%%%%%%%%%%%%%%%%%%%%%%%%%%%%%%%%%%%%
Due to both theoretical and practical interests, Reed-Solomon codes under the Hamming metric have received tremendous attention in the history of coding theory. It is natural to ask how Reed-Solomon codes perform under the insdel metric. The topic was first studied in ~\cite{WMS2004,MS2007} for small length. The first general result on Reed-Solomon codes under insdel-metric with large length was given by Tonien et al.~\cite{TS2007}. Some subsequent work on two-dimensional Reed-Solomon codes under insdel-metric was conducted by several authors \cite{DLTX2021,LT2021,CZ2021}. As mentioned above,  Reed-Solomon codes with large dimension was first studied in \cite{CST2021} where they provide a necessary and sufficient condition for which an $[n,k]$ Reed-Solomon code has insdel distance at least $2n-4k+4$. In particular, an upper bound on code alphabet size $q$ for an $[n,2]$ Reed-Solomon code is given, i.e., $q=O(n^4)$. Although it was shown in \cite{CST2021} that an $[n,k]$ Reed-Solomon code with insdel distance $d$ can achieve $d=2n-4k+4$, or equivalently, $k=\frac12\left(n-\frac d2\right)+2$, one does not know if an $[n,k]$ Reed-Solomon code with insdel distance $d$ can have insdel distance beyond $2n-4k+4$. Note that, as mentioned earlier, it was shown in \cite{CGHL2021} that any $q$-ary $[n,k]$-linear code with insdel distance $d$ must obey $d\le 2n-4k+o(n).$ This gives the third open problem.

\begin{open}
Are there any $q$-ary $[n,k]$-codes with insdel distance $d$ beyond $2n-4k+4$ for $1<k< n$.
\end{open}

\subsection{Our contributions}
Our main results can be divided into two parts. The first part gives three bounds for insdel codes, while the second part provides  constructions of nonlinear insdel codes and two-dimensional Reed-Solomon codes.

Our first bound confirms that the answer to Open Problem 1 is negative, i.e., any nontrivial codes do not achieve the insdel-metric Singleton bound.
Let $q, n,d\ge 2$ be integers. Let $I_q(n, d)$ denote the largest size $M$ of  a code of length $n$ and minimum insdel distance of at least $d$. For given $n$ and $d$, it is a central coding problem in the topic of insdel codes to determine $I_q(n, d)$. Our first upper bound shows that $I_q(n,d)\le \frac{1}{2}\left(q^{n-\frac{d}{2}+1}+q^{n-\frac{d}{2}}\right)$ for $4\le d\le 2n-2$. Furthermore, we prove that $I_q(n,d)\le q^{n-\frac{d}{2}}$ if $2q\le d\le 2n-2$.

The second upper bound of this paper gives an affirmative answer to the Open Problem 3 for the some parameter regime. Precisely speaking, we show that, if the length $n$ and dimension $k$ satisfies $n\ge \frac{k(k+1)}2+k-3$,  then every Reed-Solomon code  with length $n$ and dimension $k\ge 3$ has insdel distance at most $2n-4k+4.$
On the other hand, it was shown in \cite{CST2021} that, when the field size is sufficiently large, there are $[n,k]$ Reed-Solomon codes  with insdel distance at least $2n-4k+4$.

Our third result of the first part provides a trade-off between field size and insdel distance for Hamming-metric Singleton-optimal codes, i.e.,  every $q$-ary $(n,M)$ Singleton-optimal code $\mC$ has insdel distance at most
$2n-2k+2-2\Gd$ with $k=\log_qM$ and $\Gd\ge 2$ if
\[
q\le \left\{\begin{array}{ll}
{2^{\frac{2k-3}{\Gd-1}}}&\mbox{when $k>\frac{n+1}{3}$},\\
\left(\frac1{2(k-1)!}(n-2k+1)^{k-1}\right)^{\frac1{\Gd-1}} & \mbox{when $k\le\frac{n+1}{3}$}.
\end{array}
\right.
\]
We remark that our bound applies not only to MDS codes, which are linear Hamming-metric Singleton-optimal codes (including Reed-Solomon codes), but also to non-linear Hamming-metric Singleton-optimal codes.

Our first ``constructive" result of the second part provides an affirmative answer to Open Problem 2 by constructing nonlinear insdel codes via constant-weight $L^1$-codes. Note that our construction is not explicit in the sense that we make use of the pigeonhole principle.  Our construction shows that there are $q$-ary $(n,M)$-codes of insdel distance at least $d$ satisfying
\[\log_qM\ge n-\frac d2+1-O\left(\frac{n\log n}{\log q}\right).\]
  This means that there are nonlinear codes with logarithm of the code size arbitrarily approaching the insdel-metric Singleton bound when the code alphabet size is sufficiently large. This  non-explicit construction of insdel codes is based on constant-weight $L^1$-codes that are  introduced in this paper. We first establish a relation between constant-weight $L^1$-codes and insdel codes. Building on this relation, we construct constant-weight $L^1$-codes with reasonable parameters and subsequently give insdel codes approaching the insdel-metric Singleton bound.

The last result of this paper concerns with two-dimensional insdel Reed-Solomon codes. As two-dimensional Reed-Solomon codes of length $n$ can achieve the maximum insdel distance $2n-4$, the main problem is to find the smallest alphabet size. Via automorphism group of rational function fields, we show that two-dimensional Reed-Solomon codes of length $n$ over $\F_q$ with $q=O(n^5)$ can have insdel distance $2n-4$. Though this is worse than the one given in \cite{CST2021}, we provide a new angle to study the problem and hope that we can improve field size through this approach.

\subsection{Organization of the paper} The paper is organized as follows. In Section 2, we introduce some background on codes under the Hamming metric, insdel-metric and $L^1$-metric, as well as rational function fields and their automorphisms. In Section $3$, we derive various bounds on insdel codes: two upper bounds for the insdel codes sizes and one lower bound for the alphabet size. In the last section, we investigate constructions of insdel codes. More specifically, the first construction via constant-weight $L^1$-codes shows that nonlinear insdel codes can  arbitrarily approach the insdel-metric Singleton bound when the code alphabet size is sufficiently large, while the second one provides a construction for $2$-dimensional Reed-Solomon codes of length $n$ achieving insdel distance $2n-4$ with alphabet size $q=O(n^5) $ via automorphism group of rational function fields.

%%%%%%%%%%%%%%%%%%%%%%%%%%%%%%%%%%%%%%%%%%%%%%%%%%%%%%%%%%%%%%%%%%%%%%%%%%%%%%%%%%%%%

\section{Preliminaries}
In this section, we will introduce codes under various metrics and their relations. We also briefly discuss automorphism groups of rational function fields.
\subsection{Codes}
Codes under the Hamming metric have been well studied and much more understood than codes under other metrics. Codes under the insertion and deletion metric was first studied in \cite{VT65,L1965,L1967,T1984} in 1960s. However, codes under $L^1$-metric has very little results in literatures.
\subsubsection{Codes under Hamming  metric}
For an integer $q\ge 2$, denote by $[q]$ the set $\{1,2,\dots,q\}$. We denote by $[q]^n$ the set $\{(a_1,a_2,\dots,a_n):\; a_i\in[q]\}$. By abusing notation, we also call an element in $[q]^n$ a vector of length $n$. A $q$-ary code of length $n$ is a subset of $[q]^n$.

 The Hamming distance of two vectors $\bu,\bv\in[q]^n$, denoted by $d_H(\bu,\bv)$,  is defined to be the number of positions where $\bu$ and $\bv$ differ.  The Hamming distance of a $q$-ary $\mC\subseteq[q]^n$ is defined to be $\min\{d_H(\bu,\bv):\; \bu\neq\bv\in\mC\}$. A $q$-ary code $\mC$ with length $n$ size $|\mC|=M$ is denoted by $(n,M)_q$ or simply $(n,M)$ if there is no confusion. A well-known bound for codes under the Hamming metric is the Singleton bound. It says that a $q$-ary $(n,M)$-code with Hamming distance $d$ must obey
 \begin{equation}\label{eq:1}
 M\le q^{n-d+1}.
 \end{equation}
 A code achieving the above bound \eqref{eq:1} is called Hamming-metric Singleton-optimal (or Singleton-optimal under the Hamming metric), while a linear code achieving the above bound \eqref{eq:1} is called an MDS (maximum distance separable) code. The most important MDS codes are Reed-Solomon codes that have found various applications both theoretically and practically. Let us give a formal definition of a Reed-Solomon code below.

 Let $q$ be a prime power, let $\F_q$ be a finite field of size $q$ and $n$ be a positive integer such that $q>n.$  Let  $\Ga_1,\Ga_2,\dots,\Ga_n$ be $n$ pairwise distinct elements of $\F_q$. For an integer $k$ with $1\le k\le n$, denote by $\F_q[x]_{<k}$ the set of polynomials of degree less than $k$. The Reed-Solomon code with evaluation vector $\bGa:=(\Ga_1,\Ga_2,\dots,\Ga_n)$ and dimension $k$ is defined by
 \[\RS_\bGa(n,k):=\{(f(\Ga_1),f(\Ga_2),\dots,f(\Ga_n)):\; f(x)\in\F_q[x]_{<k}\}.\]
By counting the number of roots of polynomials, one can easily verify that $\RS_\bGa(n,k)$ is an MDS code of minimum distance $n-k+1$.
\subsubsection{Codes under insdel-metric}
The insdel distance $d_I(\bu,\bv)$ between two words $\mathbf{u}\in[q]^m,\bv\in [q]^{n}$ is the minimum number of insertions and deletions which is needed to transform $\mathbf{u}$ into $\mathbf{v}.$ It can be verified that $d_I(\cdot,\cdot)$ is indeed a distance in $[q]^n$. A common subsequence of two vectors $\mathbf{u}=(u_1,u_2,\dots,u_n),\bv=(v_1,v_2,\dots,v_n)\in [q]^{n}$ is a sequence $\ba\in [q]^\ell$ with $0\le \ell\le n$ such that there are indices $1\le i_1<i_2<\cdots<i_\ell\le n$ and $1\le j_1<j_2<\cdots<j_\ell\le n$ satisfying $(u_{i_1},u_{i_2},\dots,u_{i_\ell})=\ba=(v_{j_1},v_{j_2},\dots,v_{j_\ell})$. A longest  common subsequence of $\bu,\bv$ is a common subsequence that achieves the largest length. We denote by $\ell_\LCS(\bu,\bv)$ the length of a longest  common subsequence between $\bu$ and $\bv$, namely
\[\ell_\LCS(\bu,\bv)=\max\{0\le\ell\le n:\; \mbox{there is a common subsequence of $\bu,\bv$ of length $\ell$}\}.\]

For two vectors $\mathbf{u},\bv\in [q]^{n}$, longest common subsequences may not be unique, but $\ell_\LCS(\bu,\bv)$ is uniquely determined by the pair $(\bu,\bv)$. In fact, the insdel distance between two vectors can be calculated via the length of their longest common subsequences. Precisely, we have the following result.
\begin{lemma}[\cite{DLTX2021}]\label{lem:2.1} Let $\mathbf{u},\bv\in [q]^{n}$, then one has
\begin{equation} \label{eq:2}
d_I(\mathbf{u},\mathbf{v})= 2n-2\ell_\LCS(\bu,\bv).
\end{equation}
\end{lemma}
The insdel distance of a $q$-ary $\mC\subseteq[q]^n$ is defined to be $\min\{d_I(\bu,\bv):\; \bu\neq\bv\in\mC\}$.
An insdel code over $[q]$ of length $n$, size $M$ and minimum insdel distance $d$ is called an $(n,M)_q$-insdel code with insdel distance $d$.

A code $\mC\subseteq [q]^n$ can correct $t$ deletion errors if and only if it can correct $t$ insertion errors. Due to this equivalence, we use deletion error-correcting capability and insertion error-correcting capability interchangeably. A code $\mC\subseteq [q]^n$ of minimum insdel distance $d$ has insdel error-correcting capability up to $\left\lfloor \frac{{d}-1}{2}\right\rfloor.$

\subsubsection{Codes under $L^1$-metric}
The $L^1$-distance between two vectors $\mathbf{u}=(u_1,u_2,\dots,u_n),$ $\bv=(v_1,v_2,\dots,v_n)\in\ZZ^n$ is defined to be $d_L(\bu,\bv):=\sum_{i=1}^n|u_i-v_i|$. The $L^1$-distance is also called Manhattan distance \cite{EVY2010}. In this paper, instead of considering the whole lattice space $\ZZ^n$, we focus on the following Johnson space
\[J_n(w):=\{(a_1,a_2,\dots,a_n)\in\ZZ^n_{\ge 0}:\; \sum_{i=1}^n{a_i}=w\}.\]
A subset $\mC$ of $J_n(w)$ is called a constant weight $L^1$-code of length $n$ and weight $w$. The minimum distance of $\mC$ is defined to be
\[d_L(\mC)=\min\{d_L(\bu,\bv):\; \bu\neq\bv\in\mC\}.\]
Like codes under other metrics, we are interested in large size of constant weight $L^1$-codes given length $n$, weight $w$ and minimum distance $d$.
 \subsection{Rational function fields and their automorphisms}
Let $F=\F_q(x)$ be the rational function field. Let $\F_q[x]$ denote a polynomial ring over $\F_q$. It is a subring of $F$. For a  polynomial $f(x)$ with factorization $f(x)=\prod_{i=1}^tp_i(x)^{e_i}$, where $e_i\ge 1$ and $p_i(x)$ are polynomials of degree $m_i$, the quotient ring $\F_q[x]/(f(x))$ is isomorphic to the direct product $\prod_{i=1}^t\F_{q^{m_i}}^{e_1}$. The unit group $(\F_q[x]/(f(x)))^*$ is isomorphic to $\prod_{i=1}^t(\F_{q^{m_i}}^*\times\F_{q^{m_i}}^{e_1-1})$. Thus, it has size $\prod_{i=1}^t(q^{m_i}-1)q^{m_i(e_i-1)}$. Any polynomial $g(x)$ with $f(x)\nmid g(x)$ can be viewed as an element of $(\F_q[x]/(f(x)))^*$.

We denote by  $\Aut(F/\F_q)$ the automorphism group of $F$ over $\F_q$, i.e.,
\begin{equation}\label{eq:3}
\Aut(F/\F_q)=\{\Gs: F\rightarrow F |\; \Gs  \mbox{ is an } \F_q\mbox{-automorphism of } F\}.
\end{equation}
It is well known that any automorphism $\Gs\in \Aut(F/\F_q)$ is uniquely determined by the image $\Gs(x)$ of $x$. Moreover, it must the form
\begin{equation}\label{eq:4}
\Gs(x)=\frac{ax+b}{cx+d}
\end{equation}
for some constants $a,b,c,d\in\F_q$ with $ad-bc\neq0$ (see \cite{CC1951}).

For each pair $(a,b)\in\F_q^*\times\F_q$, we denote by $\Gs_{a,b}$ the automorphism defined by
$\Gs_{a,b}(x)={ax+b}$.
Now we consider the affine subgroup $\AGL(F/\F_q)$ of $\Aut(F/\F_q)$ given by
\begin{equation}\label{eq:5}
\AGL(F/\F_q)=\{\Gs_{a,b}:\; (a,b)\in\F_q^*\times\F_q\}.
\end{equation}
For every element $\Ga\in\F_q$, we denote by $P_\Ga$ the zero of the linear polynomial $x-\Ga$. Thus, we define $f(P_\Ga)$ to be $f(\Ga)$ for a polynomial $f(x)\in\F_q[x]$. For $\Gs_{a,b}\in \AGL(F/\F_q)$, we have $\Gs_{a,b}(x-\Ga)=\Gs_{a,b}(x)-\Ga=ax+b-\Ga=a(x-a^{-1}(\Ga-b))$. This implies that ${\Gs_{a,b}}(P_\Ga)=P_{a^{-1}(\Ga-b)}$. Let us denote by $\PP_q$ the set $\{P_\Ga:\; \Ga\in\F_q\}$.

For a polynomial $f(x)=\sum_{i=0}^mf_ix^i$ and $\Gs_{a,b}\in \AGL(F/\F_q)$, we define ${\Gs_{a,b}}(f(x))$ to be $\sum_{i=0}^mf_i(ax+b)^i$.

\begin{lemma}\label{lem:2.2} Let $\AGL(F/\F_q)$ be the affine automorphism  group of $F/\F_q$. Then we have
\begin{itemize}
\item[{\rm (i)}] $\AGL(F/\F_q)$ is $2$-transitive on the set $\PP_q$;
\item[{\rm (ii)}] An automorphism $\Gs\in \AGL(F/\F_q)$ is completely determined by its images $\Gs(P_{\Ga_1})$ and $\Gs(P_{\Ga_2})$ for any two distinct elements $\Ga_1,\Ga_2\in\F_q$;
    \item[{\rm (iii)}] $f(P_\Ga)=\Gs(f)(\Gs(P_\Ga))$ for any $f(x)\in\F_q[x]$, $\Ga\in\F_q$ and $\Gs\in \AGL(F/\F_q)$;
     \item[{\rm (iv)}] If $\Gs\in\AGL(F/\F_q)$ is not the identity map, then it has at most one fixed point in $\PP_q$.
\end{itemize}
\end{lemma}
\begin{proof} (i) Let $\Ga_1,\Ga_2,\Gb_1,\Gb_2\in\F_q$ with $\Ga_1\neq\Ga_2$ and $\Gb_1\neq\Gb_2$. Then the equation system
\begin{equation}\label{eq:6}\left\{\begin{array}{l}
\Gb_1x+y=\Ga_1\\
\Gb_2x+y=\Ga_2
\end{array}\right.\end{equation}
has a unique solution $(x,y)=(a,b)$. Apparently, $a\neq0$, otherwise we would have $\Gb_1=\Gb_2=b$.  Thus, the automorphism $\Gs_{a,b}$ maps $P_{\Ga_i}$ to $P_{\Gb_i}$ for $i=1,2$.

(ii) Let $\Gs=\Gs_{a,b}$ for some $a,b\in\F_q$. Assume that the images of $P_{\Ga_i}$ are $P_{\Gb_i}$ for $i=1,2$, respectively. Then $(a,b)$ is the unique solution of the equation system \eqref{eq:6}, i.e., $\Gs$ is completely determined by $\Gs(P_{\Ga_1})$ and $\Gs(P_{\Ga_2})$.

(iii) Let $f(x)=\sum_{i=0}^mf_ix^i$ with $f_i\in\F_q$. Then \begin{eqnarray*}\Gs(f)(\Gs(P_\Ga))&=&\left(\sum_{i=0}^mf_i(ax+b)^i\right)(P_{a^{-1}(\Ga-b)})\\
&=&\left(\sum_{i=0}^mf_i(a\times a^{-1}(\Ga-b)+b)^i\right)=f(\Ga)=f(P_\Ga).\end{eqnarray*}

(iv) Let $\Gs=\Gs_{a,b}$ be an automorphism that is not the identity, i.e., $(a,b)\neq(1,0)$. Let $P_\Ga$ be a fixed point of $\Gs$, i.e., $\Gs(P_\Ga)=P_\Ga$. Then $a^{-1}(\Ga-b)=\Ga$. This gives $(1-a)\Ga=b$. Hence, $a\neq 1$. Otherwise, we would have $b=0$. Therefore, $\Ga$ is the unique solution of $(1-a)x=b$. The proof is completed.
\end{proof}
\section{Bounds}
In this section, we derive various bounds. First of all, we derive an upper bound on the insdel code size given its length and minimum insdel distance. Secondly, we show an upper bound on the insdel minimum distance of Reed-Solomon codes.  Finally, we present a lower bound on the field size for a given Hamming-metric Singleton-optimal codes in terms of its length, dimension and minimum insdel distance.

\subsection{Upper bound on code size}
Given integers $q,n,d\ge 2$, one wonders what is the maximum  size of $q$-ary codes with length $n$ and insdel distance at least $d$. This motivates the following definition.

\begin{defn}
Given  integers $q,n,d\ge 2$, denote by $I_q(n,d)$ the largest size $M$ such that there exists an $(n,M)$-code with insdel distance at least $d$.
\end{defn}
Before deriving upper bounds on  $I_q(n,d)$, let us study relation between the Hamming and insdel metrics.

\begin{lemma}\label{lem:3.1} Let $q\ge 2$. Then for any $\bu,\bv\in[q]^n$, one has $d_I(\bu,\bv)\le 2d_H(\bu,\bv)$. In particular, we have $d_I(\mC)\le 2d_H(\mC)$ for any code $\mC\subseteq[q]^n$.
\end{lemma}
\begin{proof} %Let $\ell_\LCS(\bu,\bv)$ denote the length of a longest common subsequence between $\bu$ and $\bv$. Then
It is clear that we have  $\ell_\LCS(\bu,\bv)\ge n-d_H(\bu,\bv)$ . Thus, it follows that
\[d_I(\bu,\bv)=2n-2\ell_\LCS(\bu,\bv)\le 2n-2(n-d_H(\bu,\bv))=2d_H(\bu,\bv).\]
Now choose $\bu_0,\bv_0\in\mC$ such that $d_H(\bu_0,\bv_0)=d_H(\mC)$. Then we have $d_I(\mC)\le d_I(\bu_0,\bv_0)\le 2d_H(\bu_0,\bv_0)=2d_H(\mC)$. The proof is completed.
\end{proof}

Thus, we have the following  Singleton bound for the insdel-metric.
\begin{cor}\label{cor:3.2}
For a $q$-ary $(n,M)$-code $\mC$, one has
\begin{equation}\label{eq:7}|\mC|\le q^{n-\frac{d_I(\mC)}2+1}.\end{equation}
\end{cor}
\begin{proof}
This follows from the Singleton bound for the Hamming metric and Lemma \ref{lem:3.1}.
\end{proof}
We call the bound \eqref{eq:7}  the insdel-metric  Singleton bound.  A code achieving the bound \eqref{eq:7} is called insdel-metric Singleton-optimal.
For the Hamming metric, the Singleton bound can be achieved as long as the length is at most $q+1$ for any prime power $q$. In contrast, it can be shown that insdel Singleton bound is not achievable unless $d_I=2$ or $2n$ (see Theorem 3.6 below).

\begin{lemma}\label{lem:3.3} Let $\mC\subseteq[q]^n$ be a code with size $|\mC|=q^{n-d_H(\mC)+1}$. Then $\mC_R=[q]^{n-d_H(\mC)+1}$ for any subset $R$ of $[n]$ with $|R|=n-d_H(\mC)+1$, where $\mC_R$ stands for the projection of $\mC$ at $R$.
\end{lemma}
\begin{proof} Deleting positions of every codeword at $[n]\setminus R$ gives $\mC_R$. Apparently, $\bu_R\neq\bv_R$ whenever $\bu\neq\bv$. This is because we delete only $d_H(\mC)-1$ positions and $\mC$ has distance $d_H(\mC)$. Hence, $|\mC_R|=|\mC|=q^{n-d_H(\mC)+1}$. As $\mC_R\subseteq [q]^{n-d_H(\mC)+1}$, we get $\mC_R=[q]^{n-d_H(\mC)+1}$.
\end{proof}

\begin{lemma}\label{lem:3.4}
If $\mC\subseteq[q]^n$ is an $(n, M)$-code with $n\ge 3$ and the insdel distance $d_I(\mC)=2n-2$,  then $M\le \frac{q^2+q}{2}$, i.e., $I_q(n,2n-2)\le \frac{q^2+q}{2}$ for $n\ge 3$.
\end{lemma}
\begin{proof}
Suppose $M\ge \frac{q^2+q+1}{2}.$ Put $d=d_H(\mC)$. By Lemma \ref{lem:3.1}, we have $d=d_H(\mC)\ge \frac{d_I(\mC)}2=n-1$. If $d=n$, then $M\le q$ by the Singleton bound for the Hamming metric.  This implies that $\frac{q^2+q+1}{2}\leq M\leq q$, which is a contradiction due to the fact that for any integer $q\ge2,$ we have $q<\frac{q^2+q+1}{2}$. This contradiction forces that $d=n-1$, i.e., $2d=d_I$. Let $R_1=\{1,2\}$ and $R_2=\{2,3\}$. Then we have $|\mC_{R_1}|=|\mC_{R_2}|=|\mC|=M$. Hence,
\begin{eqnarray*}
|\mC_{R_1}\cap\mC_{R_2}|&=&|\mC_{R_1}|+|\mC_{R_2}|-|\mC_{R_1}\cup\mC_{R_2}|\\
&\ge& 2\times M-q^2\ge 2\times  \frac{q^2+q+1}{2}-q^2=q+1.
\end{eqnarray*}
It implies that there exists $(\Ga, \Gb)\in \mC_{R_1}\cap\mC_{R_2}$ with $\Ga\neq\Gb.$ Let ${\bu, \bv}$ be two codewords of $\mC$ such that ${\bu}_{R_1}={\bv}_{R_2}=(\Ga, \Gb)$, then ${\bu}=(\Ga,\Gb,*,\cdots,*)$ and ${\bv}=(*,\Ga,\Gb,\cdots,*)$. It is clear that ${\bu}\neq{\bv}$ as their second coordinates differ.  So, the length of their longest common subsequence satisfies $\ell_{\rm{LCS}}({\bu, \bv})\ge 2$, this gives $d_I(\mC)\le2(d+1)-2\times 2=2d-2.$ It is a contradiction.
\end{proof}

Lemma \ref{lem:3.4} provides a general upper bound on $I_q(n,2n-2)$. If $n\ge q+1$, then the bound can be significantly reduced. In order to obtain this improved bound, let us introduce the following map.

\begin{equation}\label{eq:8}
\phi:\; [q]^n\rightarrow J_q(n);\quad \bu=(u_1,u_2,\dots,u_n)\mapsto(a_1,a_2,\dots,a_q),
\end{equation}
where $a_i$ stand for the sizes of $\{j\in[n]:u_j=i\}$ for $i=1,2,\dots,q$.  It is clear that $\phi$ is surjective. However it is not injective. For instance, under $\phi$ for $q=3$ and $n=5$, both the ternary vector $(1,1,2,3,1)$ and $(1,3,1,1,2)$ are mapped to $(3,1,1)$.

\begin{lemma}\label{lem:3.5}
If $\mC\subseteq[q]^n$ is an $(n, M)$-code with $n\ge q+1$ and the insdel distance $d_I(\mC)=2n-2$,  then $M\le q$, i.e., $I_q(n,2n-2)\le q$ for $n\ge q+1$.
\end{lemma}
\begin{proof} Let $\mC=\{\bu_1,\bu_2,\dots,\bu_M\}$ and let $\phi(\bu_i)=(a_{i1},a_{i2},\dots,a_{iq})$. Then we form an $M\times q$ matrix
\[A:=\begin{pmatrix}
a_{11}&a_{12}&\cdots&a_{1q}\\
a_{21}&a_{22}&\cdots&a_{2q}\\
\cdots&\cdots&\cdots&\cdots\\
a_{M1}&a_{M2}&\cdots&a_{Mq}\\
\end{pmatrix}.\]
We claim that every column of $A$ has at most one entry bigger than $1$. Otherwise, we would have $i\neq j\in[M]$ and $k\in[q]$ such that $a_{ik}\ge 2$ and $a_{jk}\ge 2$. This implies that  $k$ appears at least two times in both the vectors $\bu_i$ and $\bu_j$. Therefore, $\ell_\LCS(\bu_i,\bu_j)\ge 2$. This gives $d_I(\mC)\le d_I(\bu_i,\bu_j)=2n-2\ell_\LCS(\bu_i,\bu_j)\le 2n-4$. This is a contradiction to our condition $d_I(\mC)=2n-2$. By our claim, the total number of entries bigger than $1$ is at most $q$.

On the other hand, as $n\ge q+1$, every row has at least one entry bigger than $1$. Hence, the  total number of entries bigger than $1$ is at least $M$. This gives $M\le q$.
\end{proof}
Based on the above two lemmas, we can now derives some upper bounds on $I(n,d)$.
\begin{theorem}\label{thm:3.6} We have the following results.
\begin{itemize}
\item[{\rm (i)}] $I_q(n,2)=q^n$ and $I_q(n,2n)=q$.
\item[{\rm (ii)}]
$I_q(n,d)\le \frac{1}{2}\left(q^{n-\frac{d}{2}+1}+q^{n-\frac{d}{2}}\right)$ for $4\le d\le 2n-2$.
\item[{\rm (iii)}] $I_q(n,d)\le q^{n-\frac{d}{2}}$ for $2q\le d\le 2n-2$.
\end{itemize}
% where $d_I(\mC)$ is the insdel distance of $\mC.$
\end{theorem}
\begin{proof}

Let $\mC\subseteq[q]^n$ be an $(n,M)$-code with $M=I_q(n,d)$ and insdel distance $d=d_I(\mC)$.
Put $d_H=d_H(\mC)$.

(i) By the insdel-metric Singleton bound of Corollary \ref{cor:3.2}, we have $I_q(n,2)\le q^n$ and $I_q(n,2n)\le q$. The code $[q]^n$ has insdel distance $2$, hence $I_q(n,2)\ge q^n$. The code $\{(a,a,\dots,a):\; a\in[q]\}$ has insdel distance $2n$, hence $I_q(n,2)\ge q$.

(ii) Assume that $4\le d\le 2n-2$.
Suppose
\begin{equation}\label{eq:9}
M\ge \frac{1}{2}\left(q^{n-\frac{d}{2}+1}+q^{n-\frac{d}{2}}+1\right).
\end{equation}
First of all, we claim that $d=2d_H$. Otherwise, we would have $d\le 2(d_H-1)$. By \eqref{eq:9}, we have $M\ge\frac{1}{2}\left(q^{n-d_H+2}+q^{n-d_H+1}+1\right)>q^{n-d_H+1}$. This is a contradiction to the Singleton bound for the Hamming metric.

Thus, we have $2\le d_H=\frac{d}2\le n-1$.
Define the set $\mA_i=\{(c_1,\cdots, c_n)\in\mC: c_1=i\}$, then $\mA_1,\dots,\mA_q$ are pairwise disjoint and $\bigcup_{i=1}^q\mA_i=\mC$. So, $M=|\mC|=\sum_{i=1}^q|\mA_i|\le q\max_{1\le i\le q}{|\mA_i|}$. Let $|\mA_\ell|=\max_{1\le i\le q}|\mA_i|,$ we have $|\mA_\ell|\ge \frac{M}{q}$.

 Let $\mC_1$ be the code obtained from $\mA_\ell$ by removing the first position. Then $\mC_1$ is  an $(n-1, \ge \frac{M}{q})$-code.
 It is clear that $d_H(\mC_1)\ge d_H$.

  Now, we show that $d_I(\mC_1)\ge d.$ Let $\bu$ and $\bv$ be two distinct codewords of $\mC_1\subseteq[q]^{n-1}$ satisfying
\begin{equation*}
d_I(\mC_1)=2(n-1)-2\ell_{\rm{LCS}}(\bu,\bv).
\end{equation*}
%where $\ell_{\rm{LCS}}(\bu,\bv)$ denotes the longest common subsequence between $\bu$ and $\bv$. Let $\mC$ be the code obtained from $\mC_1$ by adding the first position $\Ga_\ell$, so
As $(\Ga_\ell,\bu)$ and $(\Ga_\ell,\bv)$ are two distinct codewords of $\mC$, we have
\begin{eqnarray*}
d&\le& 2n-2\ell_{\rm{LCS}}((\Ga_\ell,\bu),(\Ga_\ell,\bv))\\
&=&2n-2(1+\ell_{\rm{LCS}}(\bu,\bv))=d_I(\mC_1).
\end{eqnarray*}
Continuing in this fashion, we can obtain a sequence $\{\mC_i\}_{i=0}^{n-d_H-1}$ with $\mC_0=\mC$ such that $\mC_i$ is a $(n-i,\ge \frac{|\mC_{i-1}|}{q^i})$-code. Furthermore,  for any $i=0,\cdots, n-d_H-1,$ we have
\begin{eqnarray*}
&d_H(\mC_i)\ge d_H(\mC_{i-1})\cdots\ge d(\mC_0)=d_H,\\
&d_I(\mC_i)\ge d_I(\mC_{i-1})\cdots\ge d_I(\mC_0)=d,\\
{\rm and}~  &|\mC_i|\ge\frac{|\mC_{i-1}|}{q}\ge\frac{|\mC_{i-2}|}{q^2}\cdots\ge\frac{|\mC_{0}|}{q^i}=\frac{M}{q^i}.
\end{eqnarray*}
Now let us consider the code $\mC_{n-d_H-1}$. It is a $(d_H+1,\ge \frac{M}{q^{n-d_H-1}})$ code with $d_H(\mC_{n-d_H-1})\ge d_H$ and $d_I(\mC_{n-d_H-1})\ge d$. First we have
\begin{eqnarray*}
|\mC_{n-d_H-1}|&\ge&\frac{M}{q^{n-d_H-1}}\ge\frac1{q^{n-d_H-1}} \left[\frac{1}{2}\left(q^{n-\frac{d}{2}+1}+q^{n-\frac{d}{2}}+1\right)\right]\\
&=&\frac1{2}\left(q^{d_H-\frac{d}{2}+2}+q^{d_H-\frac{d}{2}+1}\right)+\frac{1}{2}\left(\frac1{q^{n-d_H-1}}\right)\\
&=&\frac12(q^2+q)+\frac{1}{2}\left(\frac1{q^{n-d_H-1}}\right),
\end{eqnarray*}
where the last equality holds since $\frac{d}{2}= d_H.$ Next, we show that $d_H(\mC_{n-d_H-1})=d_H$ and $d_I(\mC_{n-d_H-1})=d.$ Suppose $d_H(\mC_{n-d_H-1})=d_H+1$, then by the Hamming-metric Singleton bound, $|\mC_{n-d_H-1}|\le q^{d_H+1-(d_H+1)+1}=q$. This  contradicts the inequality $|\mC_{n-d_H-1}|\ge \frac12(q^2+q)+\frac{1}{2}(\frac1{q^{n-d_H-1}}).$ Thus, $d_H(\mC_{n-d_H-1})=d_H$.

Now  we have $d_I(\mC_{n-d_H-1})\ge d=2d_H=2d_H(\mC_{n-d_H-1})$. On the other hand, it is true that $d_I(\mC_{n-d_H-1})\le 2d_H(\mC_{n-d_H-1})$. Thus, we get $d_I(\mC_{n-d_H-1})=2d_H(\mC_{n-d_H-1})$. Applying Lemma~\ref{lem:3.4} gives $|\mC_{n-d_H-1}|\le\frac12({q^2+q})$. This is a contradiction and the proof is completed.

(iii) Assume that $2q\le d\le 2n-2$. Suppose that $M\ge q^{n-\frac d2}+1$. Again we claim that $d=2d_H$. Otherwise, we would have $d\le 2(d_H-1)$ and this gives $M\ge q^{n-d_H+1}+1$. This is a contradiction to the Singleton bound for the Hamming metric.

By the same arguments, we can obtain a sequence $\{\mC_i\}_{i=0}^{n-d_H-1}$ with $\mC_0=\mC$ such that $\mC_i$ is a $(n-i,\ge \frac{|\mC_{i-1}|}{q^i})$-code with $d_I(\mC_i)=d$ and $d_H(\mC_i)=d_H$ for all $0\le i\le n-d_H-1$.
Applying Lemma \ref{lem:3.5} to the code $\mC_{n-d_H-1}$ gives $\frac M{q^{n-d_H-1}}\le q$, i.e., $M\le q^{n-d_H}=q^{n-\frac d2}$.
\end{proof}

\iffalse

\begin{theorem}\label{thm:3.4} Let $\mC\subset\F_q^n$ with $d_H(\mC)\ge 2$. Then one has %$d_I(\mC)\le 2(d_H(\mC)-1)$. In particular, one has
\[|\mC|< q^{n-\frac{d_I(\mC)}2+1}.\]\
\end{theorem}
\begin{proof} If $|\mC|< q^{n-{d_H(\mC)}+1}$, then by Lemma \ref{lem:3.1} we get  $|\mC|< q^{n-\frac{d_I(\mC)}2+1}$.

Now we assume that $|\mC|\ge q^{n-{d_H(\mC)}+1}$. This forces that $|\mC|= q^{n-{d_H(\mC)}+1}$ by the Singleton bound for Hamming metric. Put $d=d_H(\mC)$ and let $R_1=\{1,2,\dots,n-d-1,n-d,n-d+1\}$ and $R_2=\{1,2,\dots,n-d-1,n-d+1,n-d+2\}$. By Lemma \ref{lem:3.3}, we have $\mC_{R_1}=\mC_{R_2}=\F_q^{n-d_H(\mC)+1}$. This implies that there are two codewords of the forms $(0,0,\dots,0,0,1,*,\dots,*)$ and $(0,0,\dots,0,*,0,1,*,\dots,*)$, where the first codeword has $n-d$ consecutive 0's from the beginning, while the second codeword has $n-d-1$ consecutive 0's from the beginning. The length of an LCS of these two codewords is at least $n-d+1$. Hence, we have $d_I(\mC)\le 2n-2(n-d+1)=2(d-1)$. Thus, we have
\[|\mC|=q^{n-{d_H(\mC)}+1}\le q^{n-\frac{d_I(\mC)}2}.\]
The proof is completed.
\end{proof}

\textcolor{red}{Question: is it always true that $|\mC|\le q^{n-\frac{d_I(\mC)}2}$ for every code $\mC$ with $d_H(\mC)\ge 2$?}
\fi

\begin{rmk}{\rm
In general,  we do NOT have the upper bound $I_q(n,d)\le q^{n-\frac{d}2}$. Here is a counter-example. We label elements of $\ZZ_q$ by $\{\Ga_1,\Ga_2,\dots,\Ga_q\}$. Consider the following $q$-ary code of length $n$ with $n\le q$:
\[\mC=\{\Ga_1\cdot\b1,\Ga_2\cdot\b1,\dots,\Ga_q\cdot\b1,(\Ga_1,\Ga_2,\dots,\Ga_n)\},\] where $\b1$ stands for the all-one vector of length $n$.
Then $\mC$ is a $q$-ary $(n,q+1)$-code with insdel distance $d_I=2n-2$. Hence, $q^{n-\frac{d_I}2}=q$. However, we have  $|\mC|=q+1$.

The other nontrivial example can be found Remark \ref{rm:5}(i).
}\end{rmk}

Although it does not hold that $M\le q^{n-\frac{d_I}2}$ for a  $q$-ary $(n,M)$-nonlinear code with indel distance $d_I$, it is true if the code is linear (note that whenever we speak of linear codes, our alphabet $\F_q$ is a finite field with a prime power $q$).

\begin{cor}\label{cor:3.7} For an $[n,k]$-linear code $\mC$ with $4\le d_I(\mC)\le 2n-2n$, one has $d_I(\mC)\le 2n-2k$.
\end{cor}
\begin{proof} If $d_H(\mC)\le n-k$, then we have $d_I(\mC)\le 2d_H(\mC)\le 2n-2k$. Now assume that $d_H(\mC)= n-k+1\ge 2$. By Theorem \ref{thm:3.6},
we have $k<n-\frac{d_I(\mC)}2+1$, i.e, $k\le n-\frac{d_I(\mC)}2$. The desired result follows.
\end{proof}

\subsection{Upper bound on minimum insdel distance for Reed-Solomon codes}

In the previous section, we show that an insdel code cannot achieve the insdel-metric Singleton bound unless $d_I=2$ or $2n$. Furthermore, we will see in the next section that  insdel codes can approach the insdel-metric Singleton bound arbitrarily when the code alphabet size is sufficiently large. On the other hand, in this section, we will show that code sizes of Reed-Solomon codes are upper bounded by half of the insdel-metric Singleton bound.

For $1\le t\le n$, denote by  $S_t(n)$ the set $\{(i_1,i_2,\dots,i_t)\in[n]^t:\; 1\le i_1<i_2<\cdots<i_t\le n\}$.

\begin{lemma}\label{lem:3.8}
Let $k\ge 3$ and let $\Ga_1,\Ga_2,\dots,\Ga_n\in\F_q$ be pairwise distinct elements. If $n\ge \frac{k(k+1)}2-2$, then there exists two vectors $\bi=(i_1,\dots,i_{k-1}),\bj=(j_1,\dots,j_{k-1})\in S_{k-1}(n)$ with $\bi\neq\bj$ such that $j_\ell<i_\ell\le \frac{k(k+1)}2-2$ for $\ell=1,2,\dots,k-1$ and the matrix
\begin{equation}\label{eq:2x}\begin{pmatrix}
\Ga_{i_1}-\Ga_{j_1}&\Ga_{i_2}-\Ga_{j_2}&\cdots &\Ga_{i_{k-1}}-\Ga_{j_{k-1}}\\
\Ga_{i_1}^2-\Ga_{j_1}^2&\Ga_{i_2}^2-\Ga_{j_2}^2&\cdots &\Ga_{i_{k-1}}^2-\Ga_{j_{k-1}}^2\\
\vdots & \vdots & &\vdots\\
\Ga_{i_1}^{k-1}-\Ga_{j_1}^{k-1}&\Ga_{i_2}^{k-1}-\Ga_{j_2}^{k-1}&\cdots &\Ga_{i_{k-1}}^{k-1}-\Ga_{j_{k-1}}^{k-1}\\
\end{pmatrix}\in\F_q^{(k-1)\times (k-1)}\end{equation}
is invertible.
\end{lemma}
\begin{proof} We prove by induction and start with $k=3$.

%Since $\Ga_3+\Ga_1\neq\Ga_3+\Ga_2$, it is clear that either $\Ga_3+\Ga_1\neq\Ga_4+\Ga_3$ or $\Ga_3+\Ga_2\neq\Ga_4+\Ga_3$.Without loss of generality, we may assume that  $\Ga_3+\Ga_1\neq\Ga_4+\Ga_3$. Thus, the matrix
Since the determinant
\[\det\begin{pmatrix}
\Ga_3-\Ga_1&\Ga_4-\Ga_3\\
\Ga_3^2-\Ga_1^2&\Ga_4^2-\Ga_3^2
\end{pmatrix}=(a_3-a_1)(a_4-a_3)(a_4-a_1)\neq0,\] we may take $\bi=(3,4)$ and $\bj=(1,3)$. This completes the proof for $k=3$.

Now assume that the result is true for $k$ with $k\ge 3$, i.e., we have found $\bi=(i_1,\dots,i_{k-1}),\bj=(j_1,\dots,j_{k-1})\in S_{k-1}(n)$ with $\bi\neq\bj$ such that  $j_\ell< i_\ell\le \frac{k(k+1)}2-2$ for $\ell=1,2,\dots,k-1$ and  the matrix \eqref{eq:2x} is invertible.  Choose $j_k=j_{k-1}+1$ and consider the following determinant
\begin{equation}\label{eq:3x}F(x):=\det\begin{pmatrix}
\Ga_{i_1}-\Ga_{j_1}&\Ga_{i_2}-\Ga_{j_2}&\cdots &\Ga_{i_{k-1}}-\Ga_{j_{k-1}}&x-\Ga_{j_{k}}\\
\Ga_{i_1}^2-\Ga_{j_1}^2&\Ga_{i_2}^2-\Ga_{j_2}^2&\cdots&\Ga_{i_{k-1}}^2-\Ga_{j_{k-1}}^2 &x^2-\Ga_{j_{k}}^2\\
\vdots & \vdots & &\vdots&\vdots\\
\Ga_{i_1}^{k-1}-\Ga_{j_1}^{k-1}&\Ga_{i_2}^{k-1}-\Ga_{j_2}^{k-1}&\cdots &\Ga_{i_{k-1}}^{k-1}-\Ga_{j_{k-1}}^{k-1}&x^{k-1}-\Ga_{j_{k}}^{k-1}\\
\Ga_{i_1}^{k}-\Ga_{j_1}^{k}&\Ga_{i_2}^{k}-\Ga_{j_2}^{k}&\cdots &\Ga_{i_{k-1}}^k-\Ga_{j_{k-1}}^k&x^{k}-\Ga_{j_{k}}^{k}\\
\end{pmatrix}.\end{equation}
 It is a polynomial in $x$. By induction, the upper left diagonal $(k-1)\times(k-1)$ matrix is invertible. Hence, $F(x)$ is a nonzero polynomial of degree $k$. This implies that $F(x)$ has at most $k$ roots.  As the set $\{i_{k-1}+1,i_{k-1}+2,\dots,\frac{(k+1)(k+2)}2-2\}$ has size $\frac{(k+1)(k+2)}2-2-i_{k-1}\ge \frac{(k+1)(k+2)}2-2-\left(\frac{k(k+1)}2-2\right)=k+1$ elements, one can find  $i_k\in \{i_{k-1}+1,i_{k-1}+2,\dots,\frac{(k+1)(k+2)}2-2\}$ such that the element $\Ga_{i_k}$  is not a root of $F(x)$. Furthermore, we have $j_k=j_{k-1}+1< i_{k-1}+1\le i_k$ and  $i_k\le \frac{(k+1)(k+2)}2-2$. This completes the induction.
\end{proof}

Through Lemma ~\ref{lem:3.8}, we have the result as follow.
%%%%%%%%%%%%%%%%%%%%%%%%%%%%%%%%%%%%%%%%%%%%%%%%%%%%%%%%%%%%%%%%%%%%%
\begin{theorem}\label{thm:3.9}
If $k\ge 3$ and $n\ge \frac{k(k+1)}2+k-3$, then every Reed-Solomon code  with length $n$ and dimension $k$ has insdel distance at most $2n-4k+4.$
\end{theorem}
\begin{proof} Given a $k$-dimensional Reed-Solomon code $\RS_\bGa(n,k)$ with $n\ge \frac{k(k+1)}2+k-3$, we consider the vectors $\bi,\bj$ obtained
in Lemma \ref{lem:3.8}. As $i_{k-1}\le \frac{k(k+1)}2-2$ and $n\ge \frac{k(k+1)}2+k-3$
we can choose $i_k<i_{k+1}<\cdots<i_{2k-2}\le n$ with $i_k> i_{k-1}$ and $j_k<j_{k+1}<\cdots<j_{2k-2}\le n$ with $j_k> j_{k-1}$. Define   the matrix
\[A=\begin{pmatrix} 1& 1& \cdots &1&1&\cdots&1 \\
\Ga_{i_1}&\Ga_{i_2}&\cdots &\Ga_{i_{k-1}}&\Ga_{i_k}&\cdots&\Ga_{i_{2k-2}}\\
\vdots & \vdots & &\vdots&\vdots& &\vdots\\
\Ga_{i_1}^{k-1}&\Ga_{i_2}^{k-1}&\cdots &\Ga_{i_{k-1}}^{k-1}&\Ga_{i_k}^{k-1}&\cdots&\Ga_{i_{2k-2}}^{k-1}\\
 \Ga_{j_1}&\Ga_{j_2}&\cdots &\Ga_{j_{k-1}}&\Ga_{j_k}&\cdots&\Ga_{j_{2k-2}}\\
\vdots & \vdots & &\vdots&\vdots& &\vdots\\
\Ga_{j_1}^{k-1}&\Ga_{j_2}^{k-1}&\cdots &\Ga_{j_{k-1}}^{k-1}&\Ga_{j_k}^{k-1}&\cdots&\Ga_{j_{2k-2}}^{k-1}\\
\end{pmatrix}\in\F_q^{(2k-1)\times (2k-2)}.\]

Let $V\subseteq \F_q^{2k-1}$ be the solution space of ${\bf x}A={\bf 0}$. Since the equation has $2k-1$ variables and $2k-2$ equations, we have $\dim(V)\ge 2k-1-(2k-2)=1.$

{\bf Case 1.} $V\not\subseteq \{0\}\times\F_q^{2k-2}.$

This means that there exists a nonzero solution $(a_0,\ba,\bb)=(a_0,a_1,\cdots,a_{k-1},b_1,\cdots, b_{k-1})\in V$ of the equation system $\bx{A}=\bo$ with $a_0\neq 0$.
Let $f(x)=a_0+a_1x+\cdots+a_{k-1}x^{k-1}$ and $g(x)=-b_1x-b_2x^2-\cdots -b_{k-1}x^{k-1}$. Then $f(x)\neq g(x)$ and $(f(\Ga_{i_1}),f(\Ga_{i_2}),\cdots,f(\Ga_{i_{2k-2}}))=(g(\Ga_{j_1}),g(\Ga_{j_2}),\cdots,g(\Ga_{j_{2k-2}}))$. Hence, $\ell_\LCS(\bc_f,\bc_g)\ge 2k-2$, where $\bc_f$ stands for the codeword $(f(\Ga_1),f(\Ga_2),\dots,f(\Ga_n))$. Therefore, $d_I(\bc_f,\bc_g)\le 2n-2(2k-2)=2n-4k+4$.

{\bf Case 2.} $V\subseteq \{0\}\times\F_q^{2k-2}.$

Consider the matrix
\[B=\begin{pmatrix}
\Ga_{i_1}-\Ga_{j_1}&\Ga_{i_2}-\Ga_{j_2}&\cdots &\Ga_{i_{2k-2}}-\Ga_{j_{2k-2}}\\
\Ga_{i_1}^2-\Ga_{j_1}^2&\Ga_{i_2}^2-\Ga_{j_2}^2&\cdots &\Ga_{i_{2k-2}}^2-\Ga_{j_{2k-2}}^2\\
\vdots & \vdots & &\vdots\\
\Ga_{i_1}^{k-1}-\Ga_{j_1}^{k-1}&\Ga_{i_2}^{k-1}-\Ga_{j_2}^{k-1}&\cdots &\Ga_{i_{2k-2}}^{k-1}-\Ga_{j_{2k-2}}^{k-1}\\
\end{pmatrix}\in\F_q^{(k-1)\times (2k-2)}.\]
Choose a nonzero solution $({ 0, {\bf a, b}})\in V$. We claim $\ba\neq-\bb$. Otherwise, we would have that $\ba$ is a solution of $\bx B={\bf 0}.$  Since the first $k-1$ columns of $B$ form a $(k-1)\times (k-1)$ submatrix that is invertible, this forces that $\ba=\bo$. Thus, $({ 0, {\bf a, b}})$ is the zero solution of $\bx A={\bf 0}.$  This is a contradiction to our choice of a nonzero solution $({ 0, {\bf a, b}})\in V$.
Let $f(x)=a_1x+\cdots+a_{k-1}x^{k-1}$ and $g(x)=-b_1x-b_2x^2-\cdots -b_{k-1}x^{k-1}$. Then $f(x)\neq g(x)$ and $(f(\Ga_{i_1}),f(\Ga_{i_2}),\cdots,f(\Ga_{i_{2k-2}}))=(g(\Ga_{j_1}),g(\Ga_{j_2}),\cdots,g(\Ga_{j_{2k-2}}))$. Hence, $\ell_\LCS(\bc_f,\bc_g)\ge 2k-2$ and therefore $d_I(\bc_f,\bc_g)\le 2n-2(2k-2)=2n-4k+4$.
This completes the proof.
\end{proof}

\begin{rmk}\begin{itemize}
\item[(i)]
In the proof of Theorem \ref{thm:3.9}, as long as the matrix $B$ has rank $k-1$, we can obtain the desired result. However, in the proof we require a stronger condition, namely, the first $k-1$ columns of $B$ form an invertible $(k-1)\times (k-1)$ submatrix.  Thus, if we could get a suitable condition under which the rank of $B$ is $k-1$,  the constraint $n\ge \frac{k(k+1)}2+k-3$ could be relaxed.
\item[(ii)] In \cite{CST2021}, it is proved that when the field size $q$ is big enough, we can always have Reed-Solomon codes  with length $n$, dimension $k\ge 3$ and insdel distance equal to $2n-4k+4.$ However, one does not know if an $[n,k]$ Reed-Solomon code with insdel distance $d$ can have insdel distance beyond $2n-4k+4$.
    \item[(iii)]
Theorem \ref{thm:3.9} shows that Reed-Solomon codes with length $n$ and dimension $k>2$ can not achieve the maximum insdel distance $2n-2k$. On the other hand, $q$-ary $(n, M)$ nonlinear codes can have insdel distance at least $2n-2k+1-o(1)$ when $q$ is large enough due to a nonconstructive result given in the next section, where $k=\log_qM$ (see Remark~\ref{rm:5}(i)).
\end{itemize}
\end{rmk}

\subsection{Lower bound on field size}
In the previous section, we derived some upper bounds on the insdel distance and size of an insdel code. In this section, we focus on lower bounds on field size of Singleton-optimal codes with given length, code size and insdel distance. For convenience, in this subsection we replace the code alphabet set $[q]$ by the residue ring $\ZZ_q$ through the map $i\mapsto\bar{i}$. Thus, we can speak of supports of vectors, namely, the support $\supp(\bu)$ of a vector $\bu=(u_1,u_2,\dots,u_n)\in\ZZ_q^n$ is defined to be $\{i\in[n]:\; u_i\neq {0}\}$.

\begin{lemma}\label{lem:3.10}
Let $\mC$ be a $q$-ary Hamming-metric Singleton-optimal code with length $n$ and $k=\log_q|\mC|$. If $\mC$ contains the zero vector $\bo$, then for any $R\subseteq[n]$ with $|R|=n-k+1$, there are exactly $q-1$ codewords whose supports are exactly equal to $R$.
\end{lemma}
\begin{proof} First, let us show that there are at most $q-1$ codewords  whose supports are exactly equal to $R$. Suppose that there are at least $q$ codewords $\bc_1,\bc_2,\dots,\bc_q$ whose supports are exactly equal to $R$. Let $i\in R$. Then  there are at least two codewords, say $\bc_1,\bc_2$, that are equal at position $i$. Hence, $d_H(\bc_1,\bc_2)\le |R|-1=n-k$. This is contradiction to the fact that $\mC$ is Hamming-metric Singleton-optimal.

Now we show that there are at least $q-1$ codewords  whose supports are exactly equal to $R$.
Without loss of generality, we may assume that $R=\{k,k+1,\dots,n\}$. Put $T=\{1,2,\dots,k\}$ and define $\mA:=\{\bc_T:\; \bc\in\mC\}$. By Lemma \ref{lem:3.3} we get that  $\mA=\ZZ_q^k$. Hence, we get $\{(0,0,\dots,0,a)\in\ZZ_q^k:\; a\in\ZZ_q\}$ is a subset of $\mA$. This implies that there are $q$ codewords $\bc_0,\bc_1,\dots,\bc_{q-1}$ such that $(\bc_i)_T=(0,0,\dots,0,i)\in\ZZ_q^k$ for $0\le i\le q-1$. We claim that $\bc_0=\bo$. Otherwise, we would have $n-k+1=d_H(\mC)\le d_H(\bc_0,\bo)\le n-|T|=n-k$. This is a contradiction. Next we claim that the support of $\bc_i$ is $R$ for all $1\le i\le q-1$. This is because $n-k+1=d_H(\mC)\le d_H(\bc_i,\bo)=|\supp((\bc_i)_R)|\le n-k+1$.
This implies that $R=\supp((\bc_i)_R)=\supp(\bc_i)$ for all $1\le i\le q-1$.
\end{proof}

\begin{theorem}\label{thm:3.11} Let $q,n,k$ and $\delta$ be positive integers such that $\min\{0,2-k\}\leq \delta\leq n-k+1$ and
\begin{equation}~\label{eq:12}
\frac{q^\Gd}{q-1}\le \left\{\begin{array}{ll}
{\left\lfloor\frac{n+k+4}{2}\right\rfloor\choose k-1}&\mbox{{\rm when} $k\ge\frac{n+1}{3}$},\\
{n-k-1\choose k-1} & \mbox{{\rm when} $k\le\frac{n+1}{3}$}.
\end{array}
\right.
\end{equation}
Then every $q$-ary Hamming-metric Singleton-optimal code $\mathcal{C}\subseteq \mathbb{Z}_q^n$ of size $q^k$ has insdel distance at most $2n-2k+2-2\delta.$
\end{theorem}
\begin{proof} We may assume that $\mC$ contains the zero vector $\bo$. Otherwise, we can replace $\mC$ by $\mC-\bc$ for any $\bc\in\mC$ due to the fact that $d_I(\mC)=d_I(\mC-\bc)$.

Set
\begin{equation}\label{eq:13}
h=\left\{\begin{array}{ll}
\left\lfloor\frac{n-k}{2}\right\rfloor+2&\mbox{when $k\ge\frac{n+1}{3}$},\\
k+1 & \mbox{when $k\le\frac{n+1}{3}$}.
\end{array}
\right.\end{equation}
Let $R$ be the set $\{n-h+1,n-h+2,\dots,n\}$ with $|R|=h$. {Then $h\le n-k+1$ by definition of $h$ and constraint on $k$ given in .}
Consider the code
\[\mA:=\{\bu\in \mC:\; \wt_H(\bu)=n-k+1,\; \bu_R\in\{1,2,\dots,q-1\}^h\}.\]
By Lemma \ref{lem:3.10}, we have $|\mA|=(q-1){n-h\choose n-k+1-h}$.

We claim that $d_H(\bu_R,\bv_R)\ge 2$ for any two distinct codewords $\bu,\bv\in \mA$. Note that the Hamming distance between $\bu_{\bar{R}}$ and $\bv_{\bar{R}}$ is at most $|\bar{R}|$, where $\bar{R}=[n]\setminus R$ is the complement of $R$. Furthermore, $d_H(\bu_{\bar{R}},\bv_{\bar{R}})\le \wt_H(\bu_{\bar{R}})+\wt_H(\bv_{\bar{R}})=2(n-k+1-h)$. This gives $d_H(\bu_{\bar{R}},\bv_{\bar{R}})\le\min\{2(n-k+1-h),(n-h)\}$. Hence,
%, where $R$ is the set $\{n-\ell+1,n-\ell+2,\dots,n\}$  and $\bu_R$ stands for the projection of $\bu$ at $R$.
we have that $d_H(\bu,\bv)=d_H(\bu_{{R}},\bv_{{R}})+d_H(\bu_{\bar{R}},\bv_{\bar{R}})\le 1+\min\{2(n-k+1-h),(n-h)\}$. By the value of $h$ given in \eqref{eq:13}, we have $ 1+\min\{2(n-k+1-h),(n-h)\}<n-k+1$, so $d_H(\mC)< n-k+1$. This is a contradiction.

Thus the size of $\mA$ is $(q-1){n-h\choose n-k+1-h}=(q-1){n-h\choose k-1}$. By Theorem \ref{thm:3.6}, we have
\begin{equation}~\label{eq:14}
|\mA|=|\mA_R|= (q-1){n-h\choose k-1}< q^{h-d_I/2+1},
\end{equation}
where $d_I$ is the insdel distance of $\mA_R$. Hence,

{\bf Case $1.$} If $k\ge\frac{n+1}{3}$, by Equations~(\ref{eq:12}) and~(\ref{eq:14}) we have
\[d_I<2\left(h+1-\log_q(q-1)-\log_q{\left\lfloor\frac{n+k+4}{2}\right\rfloor\choose k-1}\right)\le 2h-2\Gd+2.\]

{\bf Case $2.$} If $k\le\frac{n+1}{3}$,  by Equations~(\ref{eq:12}) and~(\ref{eq:14}) we have
\[d_I<2\left(h+1-\log_q(q-1)-\log_q{n-(k+1)\choose k-1}\right)\le 2h-2\Gd+2.\]
Choose $\ba,\bb\in\mA$ such that $d_I=2h-2\ell_\LCS(\ba_R,\bb_R)$. Then we have $\ell_\LCS(\ba_R,\bb_R)\ge \Gd$. Hence, we have
$\ell_\LCS(\ba,\bb)\ge k-1+\ell_\LCS(\ba_R,\bb_R)\ge k-1+\Gd.$  The desired result follows.
\end{proof}

\begin{cor}\label{cor:3.12}
Let $\Gd\ge 2$. Every $q$-ary Hamming-metric Singleton-optimal code $\mC$ of length $n$ and $k=\log_q|\mC|$ has insdel distance at most
$2n-2k+2-2\Gd$ if
\begin{equation}~\label{eq:15}
q\le \left\{\begin{array}{ll}
{2^{\left\lfloor\frac{n+k+4}{2(\Gd-1)}\right\rfloor}}&\mbox{when $k\ge\frac{n+1}{3}$},\\
\left(\frac1{2(k-1)!}(n-2k+1)^{k-1}\right)^{\frac1{\Gd-1}} & \mbox{when $k\le\frac{n+1}{3}$}.
\end{array}
\right.
\end{equation}
In other words, if there exists a $q$-ary Singleton-optimal code $\mC$ of length $n$ and $k=\log_q|\mC|$ and $d_I(\mC)\ge 2n-2k+4-2\Gd$, then
\begin{equation}~\label{eq:16}
q> \left\{\begin{array}{ll}
{2^{\left\lfloor\frac{n+k+4}{2(\Gd-1)}\right\rfloor}}&\mbox{when $k\ge\frac{n+1}{3}$},\\
\left(\frac1{2(k-1)!}(n-2k+1)^{k-1}\right)^{\frac1{\Gd-1}} & \mbox{when $k\le\frac{n+1}{3}$}.
\end{array}
\right.
\end{equation}
\end{cor}
\begin{proof}
By Theorem \ref{thm:3.11}, it is sufficient to verify that the inequality \eqref{eq:12}  holds assuming that the inequality \eqref{eq:15} holds. Indeed \eqref{eq:12} holds under our assumption since\[\frac{q^\Gd}{q-1}\le 2q^{\Gd-1}, \quad 2^{\frac{n+k+4}{2}}\le {\left\lfloor\frac{n+k+4}{2}\right\rfloor\choose k-1}\quad\mbox{and}\quad\frac1{(k-1)!}(n-2k+1)^{k-1}\le{n-k-1\choose k-1}.\]
\end{proof}

\begin{rmk} It was shown in \cite{CGHL2021} that for an $[n,k]$-linear code $\mC$, one has $d_I(\mC)\le 2n-4k+o(n)$. To make $2n-2k+4-2\Gd$ to be equal to $2n-4k+o(n)$, we have to take $\Gd\approx k$, then inequalities \eqref{eq:15} and \eqref{eq:16} give trivial bounds. However, we will see from Remark~\ref{rm:5} (i)  that \eqref{eq:15} and \eqref{eq:16} are no longer trivial as there are nonlinear codes that achieve insdel-metric Singleton bound asymptotically.
\end{rmk}

\begin{cor}\label{cor:3.13} Let $\Gd\ge 2$ and $k\le \frac{n+1}{3}$. If $k=rn$ with a real $r\in(0,1)$ and $q\le 2^{\frac{(1-r)n}{\Gd-1}H_2\left(\frac{r}{1-r}\right)(1+o(1))}$, then a $q$-ary Singleton-optimal code $\mC$ of length $n$ and $k=\log_q|\mC|$ must obey
$d_I(\mC)\le 2n-2k+2-2\Gd$.
%In other words, an $[n,k]$-MDS code $\mC$ has insdel distance $d_I(\mC)> 2n-2k+2-2\Gd$ only if $q=\exp\left(\Omega(n)\right)$.}
\end{cor}
\begin{proof}
By Theorem \ref{thm:3.11}, it is sufficient to verify that the inequality \eqref{eq:12} holds under our assumption. Indeed it holds  since
\[{n-k-1\choose k-1}=2^{(1-r)nH_2\left(\frac{r}{1-r}\right)(1+o(1))}.\]
This completes the proof.
\end{proof}

\section{Constructions of insdel codes}
In this section, we will investigate constructions of insdel codes. We first present a construction of nonlinear insdel codes through constant-weight $L^1$-codes. The codes constructed in this way achieve the Singleton bound asymptotically when the code alphabet size is big enough. Secondly, via the automorphism group of rational function fields,  we provide a sufficient and necessary condition for two-dimensional Reed-Solomon codes of length $n$ to achieve insdel distance $2n-4$. Via this condition, we show the existence of  two-dimensional Reed-Solomon codes of length $n$ and  insdel distance $2n-4$ with alphabet size $q=O(n^5)$. Although, this is worse than the best known alphabet size $q=O(n^4)$, we provide a different angle to study this problem.
\subsection{Nonlinear codes}
In the equation \eqref{eq:8}, we defined the map $\phi$ from $[q]^n$ to $J_q(n)$ and note that $\phi$ is surjective and not injective. We now define a map from $J_q(n)$ to $[q]^n$.
\begin{equation}\label{eq:17}
\psi:\; J_q(n)\rightarrow[q]^n;\quad (a_1,a_2,\dots,a_q)\mapsto(\underbrace{1,1,\dots,1}_{a_1},\underbrace{2,2,\dots,2}_{a_2},\dots,\underbrace{q,q,\dots,q}_{a_q} ).
\end{equation}
Note that if $a_i=0$, then $i$ does not appear in the vector $\psi (a_1,a_2,\dots,a_q)$. It is easy to see that $\psi$ is injective. Furthermore, we have the following sequence
\[J_q(n)\xrightarrow{\psi}[q]^n\xrightarrow{\phi} J_q(n)\xrightarrow{\psi}[q]^n\]
with the composition $\phi\circ\psi$ being the identity map on $J_q(n)$.

Let us derive a relation between $L^1$-distance and insdel distance via the map $\psi$.
\begin{lemma}\label{lem:4.1}
For any $\ba, \bb\in J_q(n)$, we have $d_L(\ba,\bb)=d_I(\psi(\ba),\psi(\bb))$.
\end{lemma}
\begin{proof} Let $\ba=(a_1,a_2,\dots,a_q),\bb=(b_1,b_2,\dots,b_q)\in J_q(n)$. Then we have
$$\psi(\ba)=(\underbrace{1,\dots,1}_{a_1},\dots,\underbrace{q,\dots,q}_{a_q}),\quad\psi(\bb)=(\underbrace{1,\dots,1}_{b_1},\dots,\underbrace{q,\dots,q}_{b_q})$$ Apparently, the longest common subsequence between $\psi(\ba)$ and $\psi(\bb)$ is
\[
\ell_{\LCS}(\psi(\ba),\psi(\bb))=\sum_{i=1}^q\min\{a_i,b_i\}.
\]
By Lemma \ref{lem:2.1}, the insdel distance is
\begin{equation}\label{eq:18}
d_I(\psi(\ba),\psi(\bb))=2n-2\sum_{i=1}^q\min\{a_i,b_i\}.
\end{equation}
On the other hand, the $L^1$-distance between $\psi(\ba)$ and $\psi(\bb)$ is
\begin{eqnarray*}d_L(\psi(\ba),\psi(\bb))&=&\sum_{i=1}^q|a_i-b_i|=\sum_{i=1}^q\max\{a_i,b_i\}-\sum_{i=1}^q\min\{a_i,b_i\}\\
&=&\sum_{i=1}^q(a_i+b_i-\min\{a_i,b_i\})-\sum_{i=1}^q\min\{a_i,b_i\}\\
&=&\sum_{i=1}^q(a_i+b_i)-2\sum_{i=1}^q\min\{a_i,b_i\}=2n-2\sum_{i=1}^q\min\{a_i,b_i\}.
\end{eqnarray*}
Combining the above equality with \eqref{eq:18} gives the desired result.
\end{proof}

In view of Lemma \ref{lem:4.1}, we have the following result which shows one can obtain an insdel code from a constant weight $L^1$-code.
\begin{cor}\label{cor:4.2}
Given a $(q, M)$-constant-weight $L^1$-code $\mA\subseteq J_q(n)$ with $L^1$-distance $d=2\Gd$, then there is an $(n, M)$-code $\mC$ with insdel distance $d=2\Gd$.
\end{cor}
\begin{proof}
As the map $\psi$ defined in \eqref{eq:17} is injective, the code $\psi(\mA)$ is an $(n, M)$-code. By Lemma \ref{lem:4.1}, the insdel distance of $\psi(\mA)$ is equal to the $L^1$-distance $d=2\Gd$ of $\mA$. This completes the proof.
\end{proof}
Corollary \ref{cor:4.2} shows that one can construct insdel codes through constant weight $L^1$-codes. Thus, we are going to present a construction of constant weight $L^1$-codes.

\begin{lemma}\label{lem:4.3}
Let $r$ be a prime power bigger than  $q$. Then for any $\Gd\ge 2$, there exists a constant weight $L^1$-code in $J_q(n)$ of size \[M\ge \frac{{{n+q-1}\choose n}}{r^{\delta-2}(r-1)}\] and $L^1$-distance at least $2\Gd$.
\end{lemma}
\begin{proof} Let $\Ga_1,\Ga_2,\dots,\Ga_q,\Ga$ be $q+1$ pairwise distinct elements of $\F_r$ (this is possible as $r\ge q+1$). We define a map
\begin{equation}\label{eq:19}
\pi:\; J_q(n)\rightarrow G:=(\F_r[x]/(x-\Ga)^{\Gd-1})^*;\quad (a_1,a_2,\dots,a_q)\mapsto\prod_{i=1}^q(x-\Ga_i)^{a_i}.
\end{equation}
We claim that for any $g\in G$, the pre-image $\pi^{-1}(g)$ is a constant weight $L^1$-code of length $q$ and  $L^1$-distance at least $2\Gd$.
Define $f(x):=(x-\Ga)^{\Gd-1}$.
Take any two distinct elements $\bu=(u_1, \dots, u_q), \bv=(v_1, \dots, v_q)\in\pi^{-1}(g)$, then we have $\pi(\bu)=\pi(\bv)=g$.
This gives
 \begin{equation}\label{eq:20}
 \prod_{i=1}^q(x-\Ga_i)^{u_i}\equiv\prod_{i=1}^q(x-\Ga_i)^{v_i}\pmod {f(x)}.\end{equation}
 Let $S=\{i\in[q]:\; u_i\ge v_i\}$ and $\bar{S}$ be the complement of $S$, i.e., $\bar{S}=[q]\setminus S$. By \eqref{eq:20}, we have
 \begin{equation}\label{eq:21}
 \prod_{i\in S}(x-\Ga_i)^{u_i-v_i}\equiv  \prod_{i\in \bar{S}}(x-\Ga_i)^{v_i-u_i}\pmod {f(x)}.\end{equation}
 \iffalse
 We note the following facts:
 \begin{itemize}
 \item[(i)] $d_L(\bu,\bv)=\sum_{i=1}^q|u_i-v_i|=\sum_{i\in S}(u_i-v_i)+\sum_{i\in\bar{S}}(v_i-u_i)$;
 \item[(ii)] $\sum_{i\in S}(u_i-v_i)=\sum_{i\in\bar{S}}(v_i-u_i)$ due to the fact that $\sum_{i\in S}u_i+\sum_{i\in\bar{S}}u_i=n=\sum_{i\in S}v_i+\sum_{i\in\bar{S}}v_i$.
  \item[(iii)]   $ \prod_{i\in S}(x-\Ga_i)^{u_i-v_i}-\prod_{i\in \bar{S}}(x-\Ga_i)^{v_i-u_i}$ is a nonzero polynomial as $\bu\neq\bv$.
 \end{itemize}
 \fi
  We note the following three facts:
 \begin{itemize}
 \item[(i)] For the $L^1$-distance
 \begin{eqnarray*}
 d_L(\bu,\bv)&=&\sum_{i=1}^q|u_i-v_i|\\
 &=&\sum_{i\in S}(u_i-v_i)+\sum_{i\in\bar{S}}(v_i-u_i);
 \end{eqnarray*}
 \item[(ii)] Due to the fact that $\sum_{i\in S}u_i+\sum_{i\in\bar{S}}u_i=n=\sum_{i\in S}v_i+\sum_{i\in\bar{S}}v_i$, then
 \begin{eqnarray*}\sum_{i\in S}(u_i-v_i)=\sum_{i\in\bar{S}}(v_i-u_i)\end{eqnarray*}
  \item[(iii)]  Define $h(x):=\prod_{i\in S}(x-\Ga_i)^{u_i-v_i}-\prod_{i\in \bar{S}}(x-\Ga_i)^{v_i-u_i}$, $h(x)$  is a nonzero polynomial as $\bu\neq\bv$.
 \end{itemize}

\iffalse
Put $w=\sum_{i\in S}(u_i-v_i)=\sum_{i\in\bar{S}}(v_i-u_i)$ and
$h(x):=\prod_{i\in S}(x-\Ga_i)^{u_i-v_i}-\prod_{i\in \bar{S}}(x-\Ga_i)^{v_i-u_i}$.
 \fi
 Then $h(x)\equiv 0\pmod {f(x)}$ by \eqref{eq:21},  i.e., $f(x)|h(x)$.
Put $w=\sum_{i\in S}(u_i-v_i)=\sum_{i\in\bar{S}}(v_i-u_i)$. As both $\prod_{i\in S}(x-\Ga_i)^{u_i-v_i}$ and $\prod_{i\in \bar{S}}(x-\Ga_i)^{v_i-u_i}$ are monic polynomials of degree $w$, $h(x)$ has degree at most $w-1$. This gives $w-1\ge\deg(h(x))\ge \deg(f(x))=\delta-1$, i.e. $w\ge \Gd$. Hence, $d_L(\bu,\bv)=2w\ge 2\Gd$. This proves our claim.

By the pigeonhole principle, there exists $t\in G$ such that $|\psi^{-1}(t)|\ge \frac{|J_q(n)|}{|(\F_q[x]/f(x))^*|}$.
As $|J_q(n)|={{n+q-1}\choose n}$ and $|(\F_r[x]/(x-\Ga)^{\Gd-1})^*|= r^{\Gd-2}(r-1)$. The set $\psi^{-1}(t)$ is the code with the desired parameters.
\end{proof}
\begin{rmk}\begin{itemize}
\item[(i)] As we use the pigeonhole principle, the construction given  in Lemma \ref{lem:4.3} is not a polynomial-time one.
%On the other hand, based on our experiment result, the codes $\psi^{-1}(g)$ have almost the same size for every $g\in (\F_r[x]/(x-\Ga)^{\Gd-1})^*$.
    \item[(ii)] When $\Gd\ge 3$, we may take $r=q$ because we can take  the residue ring $\F_r[x]/(f(x))$ and its unit group for an irreducible polynomial  $f(x)$ of degree $\Gd-1$.
\end{itemize}
\end{rmk}
\begin{cor}\label{cor:4.4}
For any integers $q\ge 2$ and $n\ge\Gd\ge 2$, there exists a constant weight $L^1$-code in $J_q(n)$ of size \[M\ge \frac{{{n+q-1}\choose n}}{(2q+2)^{\delta-2}(2q+1)}\] and $L^1$-distance at least $2\Gd$.
\end{cor}
\begin{proof} As there is at least one prime between $[q+1,2(q+1)]$, we take $r$ to be a prime in the range  $[q+1,2(q+1)]$. By Lemma \ref{lem:4.3}, we have a constant weight $L^1$-code of length $q$, $L^1$-distance at least $2\Gd$ and size
\[M\ge \frac{{{n+q-1}\choose n}}{r^{\delta-2}(r-1)}\ge\frac{{{n+q-1}\choose n}}{(2q+2)^{\delta-2}(2q+1)} \]
The proof is completed.
\end{proof}

Combining Corollaries \ref{cor:4.2} with \ref{cor:4.4} gives the following result.
\begin{theorem}\label{thm:4.5}
For any integers $q\ge 2$ and $n\ge\Gd\ge 2$, there exists an $(n,M)_q$-code of size \[M\ge \frac{{{n+q-1}\choose n}}{(2q+2)^{\delta-2}(2q+1)}\] and insdel distance at least $2\Gd$.
\end{theorem}
\begin{rmk}\label{rm:5}\begin{itemize}
\item[(i)]
Let $\mC$ be the $(n,M)_q$-code with $k=\log_q M$ given in Theorem~\ref{thm:4.5}. Then
\begin{eqnarray*}
k&=&\log_qM\ge\log_q{{n+q-1}\choose n}-\log_q((2q+2)^{\delta-2}(2q+1))\\
&\ge& \log_q\frac{q^n}{n!}-(\Gd-1)-O(n)=n-\frac{d_I(\mC)}2+1-O\left(\frac{n\log n}{\log q}\right).\end{eqnarray*}
On the other hand, the Singleton bound tells us
\[k=\log_qM\le\log_qq^{n-\frac{d_I(\mC)}2+1}=n-\frac{d_I(\mC)}2+1.\]
This implies that Theorem \ref{thm:4.5} gives insdel code approaching the insdel-metric Singleton bound arbitrarily when $q$ is sufficiently large.
\item[(ii)] Let $d_I\ge 2$ be an even number. Levenshtein \cite{L2002} showed that there exists a $q$-ary $(n,M)$-code with insdel distance at least $d_I$ and
\[M\ge \frac{q^{n+d_I/2}}{\left(\sum_{i=0}^{d_I/2}{n\choose i}(q-1)^i\right)^2}.\]
For $\frac {d_I}{2}<\frac{q-1}qn$, we have
\begin{eqnarray*}
&&\log_q\left(\frac{q^{n+d_I/2}}{\left(\sum_{i=0}^{d_I/2}{n\choose i}(q-1)^i\right)^2}\right)\\
&=& n+\frac {d_I}{2}-2\left(\log_q{n\choose d_I/2}+\frac {d_I}{2}\log_q(q-1)\right)-O\left(\frac{\log n}{\log q}\right)\\
&=&n-\frac {d_I}{2}-O\left(\frac{n\log n}{\log q}\right).
\end{eqnarray*}
Let us explain the first equality above as follows.
Since ${n\choose i}(q-1)^i$ is increasing as $i$ increases for all $i<\frac {d_I}{2}$, we have   ${n\choose d_I/2}(q-1)^{d_I/2}\le\sum_{i=0}^{d_I/2}{n\choose i}(q-1)^i\le (d_I/2+1){n\choose d_I/2}(q-1)^{d_I/2}$.
 This implies that our result given in Theorem \ref{thm:4.5} is better than the one given in \cite{L2002}.

\item[(iii)] It was shown in \cite{CGHL2021} that, for an $[n,k]$-linear code, one has
\[d_I\le \frac{q-1}q(2n-4k)+o(n)\le 2n-4k+o(n),\]
i.e.,
\[k\le \frac12\left(n-\frac{d_I}2\right)+o(n).\]
Hence, linear codes can only achieve half of the bound in (ii).
\end{itemize}
\end{rmk}

\subsection{$2$-dimensional Reed-Solomon codes}
Label elements of $\F_q$ by $\Ga_1,\Ga_2,\dots,\Ga_q$.
For two vectors $\bi,\bj\in S_3(q)$, denote by $d_H(\bi,\bj)$ to be the Hamming distance between $\bi$ and $\bj$. For $\bi=(i_1,i_2,i_3)\in[q]^3$, we denote by $\bi_\bGa$ the triple $(P_{\Ga_{i_1}},P_{\Ga_{i_2}},P_{\Ga_{i_3}})$.
Furthermore, for $\Gs\in \AGL(F/\F_q)$ and $\bi=(i_1,i_2,i_3)\in[n]^3$, we denote by $\bi_\bGa^\Gs$ the triple $\bj_\bGa$ if $\Gs(P_{\Ga_{i_\ell}})=P_{\Ga_{j_\ell}}$ for $\ell=1,2,3$. Note that even if $\bi\in S_3(n)$, $\bj$ may not belong to $S_3(n)$.

\begin{lemma}\label{lem:4.6}
A $2$-dimensional Reed-Solomon code $\RS_\bGa(n,2)$  has insdel distance $2n-4$ if and only if $\bi_\bGa^\Gs\neq\bj_\bGa$ for any $\Gs\in \AGL(F/\F_q)$ and any two vectors $\bi,\bj\in S_3(n)$ with $d_H(\bi,\bj)\ge2.$
\end{lemma}
\begin{proof} Let us prove the ``only if" part first. Suppose that there exist $\Gs\in \AGL(F/\F_q)$ and two vectors $\bi,\bj\in S_3(n)$ with $d_H(\bi,\bj)\ge2$ such that $\bi_\bGa^\Gs=\bj_\bGa$. Then it is clear that $\Gs$ is  not the identity map. Now consider the polynomial $f(x)=x$ and $g(x):=\Gs(x)$. Thus, $f(x)\neq g(x)$. Furthermore, we have $f(P_{\Ga_{i_\ell}})=\Gs(f)(\Gs(P_{\Ga_{i_\ell}}))=g(P_{\Ga_{j_\ell}})$. This implies that $\ell_\LCS(\bc_f,\bc_g)\ge 3$, where $\bc_f$ stands for the vector $(f(P_{\Ga_1}),f(P_{\Ga_2}),\dots,f(P_{\Ga_n}))$. Thus, the insdel distance of $\RS_\bGa(n,2)$ is at most $2n-6$. This is a contradiction.

Now we prove the ``if" part. Suppose that the insdel distance of $\RS_\bGa(n,2)$is not $2n-4$. By Corollary \ref{cor:3.7}, any $[n,2]$-linear code has insdel distance at most $2n-4$. Thus, $\RS_\bGa(n,2)$ has insdel distance at most $2n-6$. This implies that there are   two distinct linear polynomials $f(x),g(x)$ and two vectors $\bi,\bj\in S_3(n)$ such that $f(P_{\Ga_{i_\ell}})=g(P_{\Ga_{j_\ell}})$ for $\ell=1,2,3$.

 First of all, we claim that neither $f(x)$ nor $g(x)$ is a constant. Suppose not, say for instance that $f(x)$ is a constant. If $g(x)$ is also a constant, then $\ell_\LCS(\bc_f,\bc_g)=0$. This is a contradiction. If $g(x)$ is not a constant, then $g(P_{\Ga_1}), g(P_{\Ga_2}),\dots,g(P_{\Ga_n})$ are pairwise distinct. This implies that $\ell_\LCS(\bc_f,\bc_g)\le 1$. This is a contradiction again.

 Since $f(x)-g(x)$ has degree at most one, $f(x)-g(x)$ has at most one zero. This implies that $d_H(\bi,\bj)\ge2$. As $\AGL(F/\F_q)$ is $2$-transitive on the set $\PP$, there exists an automorphism $\Gs\in \AGL(F/\F_q)$ such that $\Gs(P_{\Ga_{i_\ell}})=\Gs(P_{\Ga_{j_\ell}})$ for $\ell=1,2$. Hence, we have $g(P_{\Ga_{j_\ell}})=f(P_{\Ga_{i_\ell}})=\Gs(f)(\Gs(P_{\Ga_{i_\ell}}))=\Gs(f)(P_{\Ga_{j_\ell}})$ for $\ell=1,2$. This forces that $g=\Gs(f)$. As  $g(P_{\Ga_{j_3}})=f(P_{\Ga_{i_3}})=\Gs(f)(\Gs(P_{\Ga_{i_3}}))=g(\Gs(P_{\Ga_{i_3}}))$ and $g(x)$ is not a constant, this forces that $P_{\Ga_{j_3}}=\Gs(P_{\Ga_{i_3}})$. Therefore, $\bi_\bGa^\Gs=\bj_\bGa$.
\end{proof}
For two pairs $(\Ga_1,\Gb_1),(\Ga_2,\Gb_2)\in\F_q^2$, denote by $\Gs_{(\Ga_1,\Gb_1)\mapsto(\Ga_2,\Gb_2)}$ the unique automorphism $\Gs$ such that $\Gs(P_{\Ga_1})=P_{\Ga_2}$ and $\Gs(P_{\Gb_1})=P_{\Gb_2}$. For an automorphism $\Gs\in \AGL(F/\F_q)$, we denote by $\textsf{F}(\Gs)$ the set of fixed points of $\Gs$. By Lemma \ref{lem:2.2}(iv),  $\textsf{F}(\Gs)$ is either the empty set or has one element if $\Gs$ is not the identity map.

\begin{lemma}\label{lem:4.7} Let $n\ge 4$. If $q>\frac{n(n-1)^2(n-2)^2}4$, then
there are at least one vector $\bGa=(\Ga_1,\dots,\Ga_n)\in\F_q^n$ such that (i) $\Ga_1,\dots,\Ga_n$ are pairwise distinct; and (ii) for any $\bi,\bj\in S_3(n)$ with $d_H(\bi,\bj)\ge 2$ and $\Gs\in\AGL(F/\F_q)$, one has $\bi_\bGa^\Gs\neq\bj_\bGa$.
\end{lemma}
\begin{proof}
Let us prove by induction on $n$.

For $n=4$,
we choose any three pairwise distinct elements $\Ga_1,\Ga_2,\Ga_3\in\F_q$. Choose $P_{\Ga_4}\in\PP_q\setminus\left(\{\Gs_{(\Ga_i,\Ga_j)\mapsto(\Ga_k,\Ga_\ell)}(P_{\Ga_t})\}\cup \textsf{F}(\Gs_{(\Ga_i,\Ga_j)\mapsto(\Ga_k,\Ga_\ell)})\right)_{1\le i<j\le 3, 1\le k<\ell\le 3, 1\le t\le 3}$.  Note that this is possible since
{\small \[\left| \PP_q\setminus\left(\{\Gs_{(\Ga_i,\Ga_j)\mapsto(\Ga_k,\Ga_\ell)}(P_{\Ga_t})\}\cup \textsf{F}(\Gs_{(\Ga_i,\Ga_j)\mapsto(\Ga_k,\Ga_\ell)})\right)_{1\le i<j\le 3, 1\le k<\ell\le 3, 1\le t\le 3} \right|\ge q-{3\choose2}^2\times (1+3)>0.\]}
Then we can see that $\bi_{\bGa_1}^\Gs\neq\bj_{\bGa_1}$ for any $\Gs\in \AGL(F/\F_q)$ and any two vectors $\bi,\bj\in S_3(4)$ with $d_H(\bi,\bj)\ge2$, where $\bGa_1=(\Ga_1,\dots,\Ga_4)$.

Now assume that we have found  a vector $\{\Ga_1,\Ga_2,\dots,\Ga_{n-1}\}$ satisfying the two conditions given in this lemma. Choose
\[P_{\Ga_n}\in\PP_q\setminus\left(\{\Gs_{(\Ga_i,\Ga_j)\mapsto(\Ga_k,\Ga_\ell)}(P_{\Ga_t})\}\cup \textsf{F}(\Gs_{(\Ga_i,\Ga_j)\mapsto(\Ga_k,\Ga_\ell)})\right)_{1\le i<j\le n-1, 1\le k<\ell\le n-1, 1\le t\le n-1}.\]
 Note that this is possible since the set \[\left(\{\Gs_{(\Ga_i,\Ga_j)\mapsto(\Ga_k,\Ga_\ell)}(P_{\Ga_t})\}\cup \textsf{F}(\Gs_{(\Ga_i,\Ga_j)\mapsto(\Ga_k,\Ga_\ell)})\right)_{1\le i<j\le n-1, 1\le k<\ell\le n-1, 1\le t\le n-1}\] has at most ${{n-1}\choose 2}^2(1+n-1)=\frac{n(n-1)^2(n-2)^2}4$ elements.

Thus, we have that $\bi_\Ga^\Gs\neq\bj_\Ga$ for any $\Gs\in \AGL(F/\F_q)$ and any two vectors $\bi,\bj\in S_3{(n)}$ with $d_H(\bi,\bj)\ge2$.
\end{proof}

\begin{theorem}\label{thm:4.8} Let $n\ge 4$.
If $q>\frac{n(n-1)^2(n-2)^2}4$, then there exists an evaluation vector $\bGa$ such that the $2$-dimensional Reed-Solomon code $\RS_\bGa(n,2)$ has insdel distance $2n-4$. Furthermore, the code can be constructed in polynomial time.
\end{theorem}
\begin{proof}
The desired result follows from Lemmas \ref{lem:4.6} and  \ref{lem:4.7}.
\end{proof}

\section*{Acknowledgment}
The authors are grateful to  Venkat Guruswami for discussions during writing of the present paper.

\end{document}